\documentclass[conference]{IEEEtran}
\IEEEoverridecommandlockouts
\usepackage{cite}
\usepackage{amsmath,amssymb,amsfonts}\usepackage{bm}
\usepackage{algorithm}
\usepackage{algpseudocode}
\usepackage{graphicx}
\usepackage{textcomp}
\usepackage{xcolor}
\usepackage{optidef}
\usepackage[acronym]{glossaries}
\makeglossaries

\newacronym[type=\acronymtype]{ieee}{IEEE}{Institute of Electrical and Electronics Engineers}
\newacronym[type=\acronymtype]{ml}{ML}{Machine Learning}
\newacronym[type=\acronymtype]{lmi}{LMI}{Linear Matrix Inequality}
\newacronym[type=\acronymtype]{roa}{ROA}{Region of Attraction}
\newacronym[type=\acronymtype]{ts}{TS}{Takagi-Sugeno}
\newacronym[type=\acronymtype]{lti}{LTI}{Linear Time Invariant}
\newacronym[type=\acronymtype]{ltv}{LTV}{Linear Time Varying}
\newacronym[type=\acronymtype]{lpv}{LPV}{Linear Parameter Varying}
\newacronym[type=\acronymtype]{ldi}{LDI}{Linear Differential Inclusion}
\newacronym[type=\acronymtype]{di}{DI}{Differential Inclusion}
\newacronym[type=\acronymtype]{pldi}{PLDI}{Polytopic Linear Differential Inclusion}

\newacronym[type=\acronymtype]{lp}{LP}{Linear Programming}
\newacronym[type=\acronymtype]{qp}{QP}{Quadratic Programming}
\newacronym[type=\acronymtype]{qcqp}{QCQP}{Quadratically Constrained Quadratic Programming}
\newacronym[type=\acronymtype]{socp}{SOCP}{Second Order Cone Programming}
\newacronym[type=\acronymtype]{sdp}{SDP}{Semidefinite Programming}

\newacronym[type=\acronymtype]{sos}{SOS}{Sum of Squares}
\newacronym[type=\acronymtype]{psd}{PSD}{Positive Semidefinite}
\newacronym[type=\acronymtype]{pd}{PD}{Positive Definite}

\newacronym[type=\acronymtype]{ode}{ODE}{Ordinary Differential Equations}
\newacronym[type=\acronymtype]{are}{ARE}{Algebraic Riccati Equation}
\newacronym[type=\acronymtype]{mf}{MF}{Membership Function}
\newacronym[type=\acronymtype]{mpc}{MPC}{Model Predictive Control}
\newacronym[type=\acronymtype]{mhe}{MHE}{Moving Horizon Estimation}
\newacronym[type=\acronymtype]{clf}{CLF}{Control Lyapunov Function}
\newacronym[type=\acronymtype]{cbf}{CBF}{Control Barrier Function}

\usepackage{amsthm} % This package provides the 'lemma' environment

\newtheorem{theorem}{Theorem}

\newtheorem{definition}{Definition}

\newtheorem{remark}{Remark}
\newcommand{\innerProduct}[2]{\langle #1, #2 \rangle} % \innerProduct{x}{y}

\newcommand{\innerS}[1]{\langle #1 \rangle_{s}} % \innerS{x}

\newcommand{\matrixTwoTwo}[4]{%   \matrixTwoTwo{a}{b}{c}{d}
  \begin{bmatrix}
    #1 & #2 \\
    #3 & #4
  \end{bmatrix}%
}

\newcommand{\matrixTwoOne}[2]{%     \matrixTwoOne{a}{b}
  \begin{bmatrix}
    #1 \\
    #2
  \end{bmatrix}%
}

\def\BibTeX{{\rm B\kern-.05em{\sc i\kern-.025em b}\kern-.08em
    T\kern-.1667em\lower.7ex\hbox{E}\kern-.125emX}}
\begin{document}

% Improved Estimation of Regions of Attraction for Nonlinear Systems via Coordinate-Transformed TS Models
% \title{Improved Estimation of Regions of Attraction for Nonlinear Systems  \\ via Coordinate-Transformed TS Models \\
% }
% \title{Estimation of Regions of Attraction for Nonlinear Systems via Coordinate-Transformed TS Models
% \thanks{\textcolor{red}{XYZ}}
% }
\title{Estimation of Regions of Attraction for Nonlinear Systems via Coordinate-Transformed TS Models} %  and Quadratic Lyapunov Functions

% \title{Estimation of Regions of Attraction for Nonlinear Systems via Coordinate-Transformed TS Models
% \thanks{\textcolor{red}{This work was supported in part by the National Science Foundation (NSF) under Grant 2240512, in part by the Air Force Office of Scientific Research (AFOSR) under Grant FA9550-22-1-0476, and in part by the U.S. Department of Transportation under Grant 69A3552348327 for the CARMEN+ University Transportation Center.}}
% }

\author{
\normalsize{
\begin{tabular}{ccc}
    \begin{tabular}{c}
        Artun Sel\\
        Electrical and Computer Engineering\\
        The Ohio State University\\
        Columbus, OH 43210, USA\\
        artunsel@ieee.org
    \end{tabular}
    &
    \begin{tabular}{c}
        Mehmet Koruturk\\
        Electrical and Computer Engineering\\
        Virginia Tech\\
        Blacksburg, VA 24061, USA\\
        mkoruturk@vt.edu
    \end{tabular}
    &
    \begin{tabular}{c}
        Erdi Sayar\\
        Robotics, Artificial Intelligence \\
        and Real-time Systems\\
        Technical University of Munich \\
        Munich, Germany\\ % Federal Republic of Germany
        erdi.sayar@tum.de
    \end{tabular}
\end{tabular}
}
}

% \author{
% \IEEEauthorblockN{Artun Sel}
% \IEEEauthorblockA{\textit{Electrical and Computer Engineering} \\
% \textit{The Ohio State University}\\
% Columbus, OH 43210, USA \\
% artunsel@ieee.org}
% \and
% \IEEEauthorblockN{Mehmet Koruturk}
% \IEEEauthorblockA{\textit{Electrical and Computer Engineering} \\
% \textit{Virginia Tech}\\
% Blacksburg, VA 24061, USA \\
% artunsel@ieee.org}
% \and
% \IEEEauthorblockN{Artun Sel}
% \IEEEauthorblockA{\textit{Electrical and Computer Engineering} \\
% \textit{The Ohio State University}\\
% Columbus, OH 43210, USA \\
% artunsel@ieee.org}
% }

% \IEEEauthorblockN{Uygar Gunes}
% \IEEEauthorblockA{\textit{Electrical and Electronics Engineering} \\
% \textit{TOBB ETU}\\ %  University of Economics and Technology
% Ankara, The Republic of Turkey \\
% uygargunes@etu.edu.tr}

% \IEEEauthorblockN{Cosku Kasnakoglu}
% \IEEEauthorblockA{\textit{Electrical and Electronics Engineering} \\
% \textit{TOBB ETU}\\ %  University of Economics and Technology
% Ankara, The Republic of Turkey \\
% kasnakoglu@etu.edu.tr}

\maketitle

% ---- BEGIN ABSTRACT ----
\begin{abstract}
This paper presents a novel method for estimating larger \glspl{roa} for continuous-time nonlinear systems modeled via the \gls{ts} framework. While classical approaches rely on a single TS representation derived from the original nonlinear system to compute an \gls{roa} using Lyapunov-based analysis, the proposed method enhances this process through a systematic coordinate transformation strategy. Specifically, we construct multiple TS models, each obtained from the original nonlinear system under a distinct linear coordinate transformation. Each transformed system yields a local \gls{roa} estimate, and the overall \gls{roa} is taken as the union of these individual estimates. This strategy leverages the variability introduced by the transformations to reduce conservatism and expand the certified stable region. Numerical examples demonstrate that this approach consistently provides larger \glspl{roa} compared to conventional single-model \gls{ts}-based techniques, highlighting its effectiveness and potential for improved nonlinear stability analysis.

\end{abstract}

      % or \input{sections/abstract} if it lives in sections/
% ---- END ABSTRACT ----

% ---- BEGIN SECTION-1 ----

\section{Introduction}  \label{sec:Introduction}

% Paragraph-1
The analysis and control of nonlinear systems remains a fundamental challenge in modern control theory \cite{Kurkcu18:Disturbance,Kurkcu18:DisturbanceUncertainty,Kurkcu19:Robust}. 
Especially with the recent developments in autonomous aerial vehicles with their applications in remote sensing \cite{Hind22:Assessment} and cyber security issues in communications \cite{Etcibasi24:Coverage,Mohammadi25:Detection,Mohammadi25:GPS,Mohammadi25:DetectionICAIC,Ahmari25:AData,Ahmari25:Evaluating} and energy systems
\cite{Agrawal25:Performance,Himmelstoss25:Floating,Koch25:Qualitative,Albasheri25:Energy,Ekawita25:Investigating,Pham25:Potential}, the safety considerations are becoming more relevant. Additionally, the stability concerns may appear in computational models \cite{Sel25:LLMs,Jin24:Democratizing,Sel24:Skin,Sel24:Algorithm} and optimization algorithms \cite{Gu25:Safe,Gu24:Balance,Jin23:OnSolution,Sel23:Learning,Tawaha23:Decision,Khattar23:CMDPWithinOnline,Coskun22:Magnetic} that are especially important in estimation problems related to aerospace industry \cite{Pourtakdoust23:Advanced,Nasihati22:Satellite,Nasihati21:OnLine,Nasihati19:Autonomous}. Among various methodologies developed to address this, the \gls{ts} fuzzy modeling framework has proven to be particularly effective. It provides a systematic way to approximate nonlinear dynamical systems through a convex combination of linear submodels, enabling the application of powerful \gls{lmi} techniques for controller design and stability verification. When constructed via the sector nonlinearity approach, \gls{ts} models preserve local equivalence to the original nonlinear dynamics within a specified compact domain, making them a practical tool for analyzing local stability properties \cite{Guerra18:LMI_Based_Observer}.

A key application of \gls{ts} modeling lies in the estimation of the \gls{roa} around an equilibrium point, often the origin. Traditional \gls{ts}-based stability analysis typically involves constructing a single \gls{ts} model of the nonlinear system and verifying local or global stability conditions using a common Lyapunov function. These methods, while computationally tractable, often yield conservative \gls{roa} estimates due to their reliance on a fixed coordinate representation and restrictive Lyapunov structures \cite{LopezEstrada17:LMI_based_fault}.

To mitigate these limitations, various enhancements have been proposed, such as the use of nonquadratic Lyapunov functions, piecewise or parameter-dependent functions, and membership-function-dependent relaxations. However, the potential of exploiting multiple coordinate representations of the system to enrich the stability analysis remains largely unexplored \cite{DellaRossa20:MaxMinLyapunov}.

In this work, we introduce a novel approach to expand the estimated \gls{roa} by systematically transforming the original system through a family of linear coordinate changes. Each transformed system yields a distinct \gls{ts} model, and an individual \gls{roa} is computed via standard \gls{lmi}-based techniques. The final \gls{roa} is then obtained as the union of these individual estimates. This framework leverages the structural variation introduced by coordinate transformations to reduce the conservativeness inherent in single-model methods \cite{Sotiropoulos18:Causality,Igarashi20:ARobust}.

We demonstrate through a representative example that our method produces significantly larger \gls{roa} estimates compared to conventional approaches, while preserving computational tractability. The proposed method offers a promising direction for extending the applicability and effectiveness of \gls{ts}-based stability analysis for nonlinear systems.

The rest of the paper is organized as follows: Section~\ref{sec:Preliminaries} introduces the preliminaries.
Section~\ref{sec:ROA_TS} reviews the standard \gls{ts} modeling and \gls{roa} computation framework. Section~\ref{sec:ROA_CoC} presents the proposed method based on coordinate transformation and \gls{roa} union. Section~\ref{sec:CONCLUSIONS} concludes the paper and discusses future directions.

% ---- END SECTION-1 ----

% ---- BEGIN SECTION-2 ----

\section{Preliminaries}  \label{sec:Preliminaries}

% =======================================================================

% \begin{table}[htbp]     \label{tab:symbol_definitions}
%     \centering
%     \caption{Acronyms}
%     \begin{tabular}{|p{1cm}|p{6cm}|}
%         \hline
%         \textbf{Symbol} & \textbf{Definition} \\
%         \hline
%         $LMI$ & Linear Matrix Inequality \\
%         \hline
%         $SDP$ & Semi Definite Program \\
%         \hline
%         $LP$ & Linear Program \\
%         \hline
%         $QP$ & Quadratic Program \\
%         \hline
%         $QCQP$ & Quadratically Constrained Quadratic Program \\
%         \hline
%         $SOCP$ & Second Order Cone Program \\
%         \hline
%         $ODE$ & Ordinary Differential Equations \\
%         \hline
%         $LTI$ & Linear Time Invariant \\
%         \hline
%         $LTV$ & Linear Time Varying \\
%         \hline
%         $LPV$ & Linear Parameter Varying \\
%         \hline
%         $ARE$ & Algebraic Riccati Equation \\
%         \hline
%         $ROA$ & Region of Attraction \\
%         \hline
%         $MF$ & Membership Function \\
%         \hline
%         $TS$ & Takagi-Sugeno \\
%         \hline
%     \end{tabular}
% \end{table}

% Table~\ref{tab:symbol_definitions} provides a list of acronyms used throughout this study and
Table~\ref{tab:notation} summarizes the mathematical notation used throughout this paper for quick reference.

\begin{table}[htbp]
\centering
\caption{Mathematical Notation}
\label{tab:notation}
\begin{tabular}{|p{1cm}|p{6cm}|}
\hline
\textbf{Symbol} & \textbf{Description} \\
\hline
$\mathbb{R}^n$ & $n$-dimensional Euclidean space \\
\hline
$\mathbb{R}^{m \times n}$ & Set of all $m \times n$ real matrices \\
\hline
$\mathbb{S}^n$ & Space of $n \times n$ symmetric matrices \\
\hline
$\mathbb{S}^n_+$ & Cone of $n \times n$ symmetric positive semidefinite matrices \\
\hline
$\mathbb{S}^n_{++}$ & Cone of $n \times n$ symmetric positive definite matrices \\
\hline
$X \succeq 0$ & Matrix $X$ is positive semidefinite \\
\hline
$X \succ 0$ & Matrix $X$ is positive definite \\
\hline
$X \succeq Y$ & Matrix $X - Y$ is positive semidefinite \\
\hline
$\text{tr}(X)$ & Trace of matrix $X$ (sum of diagonal elements) \\
\hline
$\text{det}(X)$ & Determinant of matrix $X$\\
\hline
$\text{diag}(X)$ & diagonal of matrix $X$\\
\hline
$X^{\top}$ & Transpose of matrix $X$ \\
\hline
$x^{\top}$ & Transpose of vector $x$ \\
\hline
$A \otimes B$ & Kronecker product of matrices $A$ and $B$ \\
\hline
$\innerProduct{A}{B}$ & Inner product of matrices, defined as $\text{tr}(A^{\top} B)$ \\
\hline
$\lambda_i(X)$ & The $i$-th eigenvalue of matrix $X$ \\
\hline
$v(P)$ & Optimal value of optimization problem $(P)$ \\
\hline
$\text{conv}(S)$ & Convex hull of set $S$ \\
\hline
$ \innerS{X} $ & $X + {X^{\top}}$ \\
\hline
\end{tabular}
\end{table}

% =======================================================================

\subsection{\textbf{Definitions for Mathematical Optimization Related Concepts}}
\subsubsection{\textbf{Linear Matrix Inequality}}
% =======================================================================
An \gls{lmi} is a constraint of the form:
\begin{equation}
F(x) = F_0 + \sum_{i=1}^m x_i F_i \succeq 0
\end{equation}
where $x = (x_1, x_2, \ldots, x_m)$ is the variable vector, the symmetric matrices $F_0, F_1, \ldots, F_m \in \mathbb{R}^{n \times n}$ are given, and the inequality symbol $\succeq 0$ denotes positive semidefiniteness. An \gls{lmi} defines a convex constraint on $x$. Multiple \glspl{lmi} can be expressed as a single \gls{lmi} using block-diagonal structure. \glspl{lmi} are central to semidefinite programming, where they serve as constraints in convex optimization problems that can be solved efficiently using interior-point methods \cite{boyd2004convex}.

% Their prevalence in systems and control engineering stems from their ability to represent Lyapunov stability conditions, performance criteria, and robust control problems in a computationally tractable framework.
% =======================================================================

\subsubsection{\textbf{Semi Definite Programming}}
% =======================================================================
\gls{sdp} is a powerful optimization framework that extends linear programming to the cone of \gls{psd} matrices. Formally, an \gls{sdp} can be expressed as:
\begin{equation}
\begin{aligned}
\text{minimize} \quad & \innerProduct{C}{X} \\
\text{subject to} \quad & \innerProduct{A_i}{X} \leq b_i, \quad i = 1, 2, \ldots, m \\
& X \succeq 0
\end{aligned}
\end{equation}
% \begin{align}
% \text{minimize} \quad & \innerProduct{C}{X} \\
% \text{subject to} \quad & \innerProduct{A_i}{X} \leq b_i, \quad i = 1, 2, \ldots, m \\
% & X \succeq 0
% \end{align}
\noindent where $X, C, A_1, \ldots, A_m$ are symmetric $n \times n$ matrices, $b_1, \ldots, b_m$ are scalars, and $X \succeq 0$ denotes that $X$ is \gls{psd}. The constraint $X \succeq 0$ means that all eigenvalues of $X$ are non-negative \cite{Tedrake10:LQR}.

\gls{sdp} has significant applications in control theory, combinatorial optimization, and machine learning. It generalizes several optimization problems, including linear and quadratic programming, while remaining solvable in polynomial time using interior-point methods. What makes \gls{sdp} particularly valuable is its ability to provide tight relaxations for many NP-hard problems, often yielding high-quality approximate solutions \cite{Recht10:Guaranteed}.

% =======================================================================

\subsubsection{\textbf{Relaxation in Mathematical Optimization}}
% =======================================================================
Relaxation is a fundamental technique in mathematical optimization where a complex problem is approximated by a simpler one that is more tractable to solve. Formally, given an optimization problem $(P)$:
\begin{align}
(P): \quad \min_{x \in \mathcal{F}} \quad f(x)
\end{align}
\noindent a relaxation $(R)$ of $(P)$ is another optimization problem:
\begin{align}
(R): \quad \min_{x \in \mathcal{F}'} \quad g(x)
\end{align}
\noindent where $\mathcal{F} \subseteq \mathcal{F}'$ and $g(x) \leq f(x)$ for all $x \in \mathcal{F}$.
The key property of a relaxation is that it provides a lower bound on the optimal value of the original problem. That is, if $v(P)$ and $v(R)$ denote the optimal values of problems $(P)$ and $(R)$ respectively, then $v(R) \leq v(P)$.
Relaxations are particularly useful for addressing difficult optimization problems, especially those that are non-convex or NP-hard. By replacing hard constraints with easier ones or approximating non-convex objective functions with convex ones, relaxations transform intractable problems into ones that can be efficiently solved. The quality of a relaxation is measured by the gap between the original problem's optimal value and the relaxation's optimal value—a smaller gap indicates a tighter relaxation \cite{Nedic10:Constrained}.

% =======================================================================

\subsubsection{\textbf{\gls{sdp} Relaxation}}
% =======================================================================
\gls{sdp} relaxation is a powerful technique that transforms difficult optimization problems, particularly non-convex quadratic problems, into tractable \glspl{sdp}. Given a \gls{qcqp} of the form:
\begin{equation}
\begin{aligned}
\min_{x \in \mathbb{R}^n} \quad & {x^{\top}} P_0 x + {{q_0}^{\top}} x + r_0 \\
\textrm{subject to} \quad & {{x}^{\top}} P_i x + {{q_i}^{\top}} x + r_i \leq 0, \\ & \forall i \in \{ 1, 2, \ldots, m \}
\\ & X \succeq 0
\end{aligned}
\end{equation}
% \begin{align}
% \min_{x \in \mathbb{R}^n} \quad & {x^{\top}} P_0 x + {{q_0}^{\top}} x + r_0 \\
% \textrm{subject to} \quad & {{x}^{\top}} P_i x + {{q_i}^{\top}} x + r_i \leq 0, \quad i = 1, 2, \ldots, m
% \end{align}

\noindent the \gls{sdp} relaxation proceeds by introducing a matrix variable $X = x{{x}^{\top}}$ and relaxing this non-convex equality constraint to $X \succeq x{{x}^{\top}}$. By the Schur complement condition, this is equivalent to requiring that:
\begin{align}
\begin{pmatrix} 
1 & {{x}^{\top}} \\ 
x & X
\end{pmatrix} \succeq 0
\end{align}
\noindent The resulting \gls{sdp} relaxation is:
\begin{equation}
\begin{aligned}
\min_{x \in \mathbb{R}^n, X \in \mathbb{S}^n} \quad & \text{tr}(P_0 X) + {{q_0}^{\top}} x + r_0 \\
\textrm{subject to} \quad & \text{tr}(P_i X) + {{q_i}^{\top}} x + r_i \leq 0, \\ & \forall i \in \{ 1, 2, \ldots, m \}
\\ & 
\begin{pmatrix} 
1 & {{x}^{\top}} \\ 
x & X
\end{pmatrix} \succeq 0
\end{aligned}
\end{equation}
% \begin{align}
% \min_{x \in \mathbb{R}^n, X \in \mathbb{S}^n} \quad & \text{tr}(P_0 X) + {{q_0}^{\top}} x + r_0 \\
% \textrm{subject to} \quad & \text{tr}(P_i X) + {{q_i}^{\top}} x + r_i \leq 0, \quad i = 1, 2, \ldots, m \\
% & \begin{pmatrix} 
% 1 & {{x}^{\top}} \\ 
% x & X
% \end{pmatrix} \succeq 0
% \end{align}
\noindent where $\mathbb{S}^n$ denotes the space of $n \times n$ symmetric matrices \cite{Kuindersma14:AnEfficiently}.

% \gls{sdp} relaxations have gained prominence due to their ability to provide tight approximations for many hard problems with theoretical guarantees. Notable applications include the MAX-CUT problem, where the Goemans-Williamson algorithm uses \gls{sdp} relaxation to achieve an approximation ratio of 0.878, and various problems in control theory, combinatorial optimization, and machine learning.
% =======================================================================

\subsection{\textbf{Definitions for Applications of the Methods}}

\subsubsection{\textbf{Dynamical System}}
A dynamical system is a mathematical model that describes the temporal evolution of a physical system's state according to a fixed rule. Formally, it consists of a state space $\mathcal{X}$ and a function $f: \mathcal{X} \times \mathbb{R} \to \mathcal{X}$ that specifies how the state $x(t) \in \mathcal{X}$ evolves over time. For continuous-time systems, this evolution is typically described by a differential equation:
\begin{equation}
\dot{x}(t) = f(x(t), t)
\end{equation}
\noindent where $\dot{x}(t)$ represents the derivative of the state with respect to time. Dynamical systems can be classified as linear or nonlinear, time-invariant or time-varying, and autonomous or non-autonomous. The analysis of these systems focuses on characterizing their stability properties, equilibrium points, limit cycles, and other qualitative behaviors \cite{Duan13:LMIs}.

\begin{remark}
    In this study, we primarily focus on polynomial-type nonlinear systems that are time-invariant. For analysis problems, we specifically examine autonomous systems that possess these properties.
\end{remark}

\subsubsection{\textbf{Stability Definitions}} % [GLOBAL-LOCAL] [Asymptotical-Exponential]
% =======================================================================
Consider a nonlinear dynamical system described by:
\begin{equation}
\label{eq:autonomous_NL_DYN_SYS}
\dot{x} = f(x),\quad x(0) = x_0
\end{equation}
where $x \in \mathbb{R}^n$ is the state vector, $f: \mathbb{R}^n \rightarrow \mathbb{R}^n$ is a locally Lipschitz function, and $f(0) = 0$ (i.e., the origin is an equilibrium point). We define various notions of stability as follows:

\begin{definition}[\textbf{Stability in the Sense of Lyapunov}]
The equilibrium point $x = 0$ is stable in the sense of Lyapunov if, for any $\epsilon > 0$, there exists a $\delta > 0$ such that:
\begin{equation}
\|x(0)\| < \delta \implies \|x(t)\| < \epsilon, \quad \forall t \geq 0
\end{equation}
\end{definition}

\begin{definition}[\textbf{Asymptotic Stability}]
The equilibrium point $x = 0$ is asymptotically stable if it is stable in the sense of Lyapunov and there exists a $\delta > 0$ such that:
\begin{equation}
\|x(0)\| < \delta \implies \lim_{t \rightarrow \infty} \|x(t)\| = 0
\end{equation}
\end{definition}

\begin{definition}[\textbf{Exponential Stability}]
The equilibrium point $x = 0$ is exponentially stable if there exist positive constants $\alpha$, $\beta$, and $\delta$ such that:
\begin{equation}
\|x(0)\| < \delta \implies \|x(t)\| \leq \alpha \|x(0)\| e^{-\beta t}, \quad \forall t \geq 0
\end{equation}
\end{definition}

\begin{definition}[\textbf{Global Stability}]
A stability property (Lyapunov, asymptotic, or exponential) is said to be global if it holds for any initial state $x(0) \in \mathbb{R}^n$ (i.e., $\delta \rightarrow \infty$).
\end{definition}

\begin{definition}[\textbf{Local Stability}]
A stability property is said to be local if it holds only for initial states within some bounded neighborhood of the equilibrium (i.e., $\|x(0)\| < \delta$ for some finite $\delta > 0$).
\end{definition}
% =======================================================================
\subsubsection{\textbf{Analysis Problem in LTI systems}} % ANALYSIS
% \textcolor{red}{1. asymp stab} \\
% \textcolor{red}{2. robust stab} \\
% \textcolor{red}{3. quad stab} \\
% =======================================================================
The analysis problem in LTI systems involves determining system properties such as stability, performance, and robustness given a complete system description but we mainly refer to stability in this report. Consider a continuous-time LTI system:
\begin{equation}
\dot{x}(t) = Ax(t)
\end{equation}
where $x(t) \in \mathbb{R}^n$ is the state vector and $A \in \mathbb{R}^{n \times n}$ is the system matrix.

The stability analysis problem can be formulated as verifying whether all eigenvalues of $A$ have negative real parts. Using Lyapunov theory, this is equivalent to determining whether there exists a positive definite matrix $P$ such that:
\begin{equation}
% {{A}^{\top}}P + PA \prec 0
-\innerS{PA} \in \mathbb{S}_{++}
\end{equation}
This verification can be expressed as a convex feasibility problem:
\begin{equation}
\begin{aligned}
\text{find} \quad & P \\
\text{subject to} \quad & P  \in \mathbb{S}_{++} \\
& -\innerS{P A}  \in \mathbb{S}_{++}
\end{aligned}    
\end{equation}
This formulation presents a significant advantage: the problem can be transformed into a LMI constraint, making it convex and solvable in polynomial time. The framework extends naturally to other analysis problems including $\mathcal{H}_2$ and $\mathcal{H}_{\infty}$ performance assessment, input-to-state stability, and robust stability against model uncertainties, with each case expressible as an SDP with appropriate LMI constraints.
% =======================================================================
\subsubsection{\textbf{Synthesis Problem in LTI systems}} % SYNTHESIS
% =======================================================================
The synthesis problem in LTI systems involves designing controllers that achieve desired closed-loop properties such as stability, performance, and robustness but we mainly
refer to the stabilizing controller design problem in this report.
Consider a controllable LTI system:
\begin{equation}
\dot{x}(t) = Ax(t) + Bu(t)
\end{equation}
where $x(t) \in \mathbb{R}^n$ is the state vector, $u(t) \in \mathbb{R}^m$ is the control input, $A \in \mathbb{R}^{n \times n}$ is the system matrix, and $B \in \mathbb{R}^{n \times m}$ is the input matrix. 

For a state feedback control law $u(t) = Kx(t)$, the closed-loop system becomes:
\begin{equation}
\dot{x}(t) = (A + BK)x(t)
\end{equation}

A typical synthesis problem aims to find a stabilizing feedback gain $K$ that minimizes a performance criterion, such as an $\mathcal{H}_2$ or $\mathcal{H}_{\infty}$ norm. This can be formulated as:
\begin{equation}
\begin{aligned}
\text{minimize} \quad & \text{Performance Measure}(A + BK) \\
\text{subject to} \quad & (A + BK) \text{ is stable}
\end{aligned}
\end{equation}
If there is no performance criterion to be optimized, then the resulting problem is known as stabilizing controller design problem and it can be given by the following feasibility problem:
\begin{equation}
\begin{aligned}
\text{find} \quad & (P,K) \\
\text{subject to} \quad P & \in \mathbb{S}_{++} \\
-\innerS{P A_{c}} & \in \mathbb{S}_{++} \\
A_{c} & = A+BK
\end{aligned}
\end{equation}
While this problem appears non-convex due to the product term $PBK$, it can be transformed into a convex problem through a change of variables. By defining $X = P^{-1}$ and $M = KX$ where 
$P \in \mathbb{S}_{++}$,
% $P \succ 0$
the stability constraint becomes:
\begin{equation}
\innerS{AX + BM} \in \mathbb{S}_{++}
\end{equation}
which results in the following problem:
\begin{equation}
\begin{aligned}
\text{find} \quad & (X,M) \\
\text{subject to} \quad X &  \in \mathbb{S}_{++} \\
-\innerS{AX + BM} &  \in \mathbb{S}_{++}
\end{aligned}
\end{equation}
where the $(P,K)$ terms are given by $(X^{-1},M{X^{-1}})$.
\begin{remark}
    It should be noted that since $P \in \mathbb{S}_{++}^{n}$, it is an invertible matrix and $K$ can be computed after computing the $(X,M)$ pair.
\end{remark}
% =======================================================================
\subsubsection{\textbf{Lyapunov Function and Stability of Nonlinear Dynamical Systems}}
% LYAPUNOV THEOREMS
\begin{theorem}
For \eqref{eq:autonomous_NL_DYN_SYS}, where the origin is the equilibrium point,
If $\exists V(x): \mathcal{D} \rightarrow \mathbb{R}$ \textit{continuously differentiable} such that
\begin{align*}
    V(0) &= 0, \\
    \dot{V}(0) &= 0, \\
    V(x) &> 0, \quad \forall x \in \mathcal{D} \setminus \{0\}, \\
    \dot{V}(x) &< 0, \quad \forall x \in \mathcal{D} \setminus \{0\},
\end{align*}
then $\mathbf{x} = 0$ is \textit{asymptotically stable}; if, in addition, $\mathcal{D} = \mathbb{R}^n$ and $V(x)$ is \textit{radially unbounded}, then $\mathbf{x} = 0$ is \textit{globally asymptotically stable}.
If $V(x)$ satisfies the constraints in $\Omega \subseteq \mathcal{D}$, then $\mathbf{x} = 0$ is \textit{locally asymptotically stable} and $\Omega$ is called \textit{ROA}.
\end{theorem}
% =======================================================================
% =======================================================================
\subsubsection{\textbf{Convex Optimization Based Nonlinear Dynamical System Analysis}}
% =======================================================================
\begin{definition}
A \textit{convex model for autonomous systems} is a set of first-order \glspl{ode} whose right-hand side can be expressed as a convex sum of vector fields (models), namely
\begin{equation}
\label{eq:convex_model_autonomous}
    \dot{\mathbf{x}}(t) = \sum_{i=1}^{r} h_i(\mathbf{z}) f_i(\mathbf{x})
\end{equation}
where $r \in \mathbb{N}$ represents the number of models in the convex sum, $\mathbf{z} \in \mathbb{R}^p$ may depend on the state $\mathbf{x}$, time $t$, disturbances or exogenous parameters, and the nonlinear functions $h_i(\mathbf{z})$, $i \in \{1, 2, \ldots, r\}$ hold the \textit{convex sum property} in a compact set $\mathcal{C} \subseteq \mathbb{R}^p$, i.e.:
\begin{equation}
\label{eq:convex_sum_property}
    \sum_{i=1}^{r} h_i(\mathbf{z}) = 1, \quad 0 \leq h_i(\mathbf{z}) \leq 1, \quad \forall \mathbf{z} \in \mathcal{C}.
\end{equation}
and the special case is given by
\begin{equation}
\label{eq:convex_sum_Ai_autonomous}
    \dot{\mathbf{x}}(t) = \sum_{i=1}^{r} h_i(\mathbf{z}) A_i \mathbf{x}(t)
\end{equation}
\end{definition}
%
% =======================================================================
\begin{definition}
A \textit{convex model for nonautonomous systems} is a set of first-order \glspl{ode} whose right-hand side can be expressed as a convex sum of vector fields (models), namely
\begin{equation}
\label{eq:convex_model_w_input}
    \dot{\mathbf{x}}(t) = \sum_{i=1}^{r} h_i(\mathbf{z}) f_i(\mathbf{x}, \mathbf{u})
\end{equation}
where $r \in \mathbb{N}$ represents the number of models in the convex sum, $\mathbf{z} \in \mathbb{R}^p$ may depend on the state $\mathbf{x}$, input $\mathbf{u}$, time $t$, disturbances or exogenous parameters, and the nonlinear functions $h_i(\mathbf{z})$, $i \in \{1, 2, \ldots, r\}$ hold the \textit{convex sum property} in a compact set $\mathcal{C} \subseteq \mathbb{R}^p$, i.e.:
\begin{equation}
\label{eq:convex_sum_property}
    \sum_{i=1}^{r} h_i(\mathbf{z}) = 1, \quad 0 \leq h_i(\mathbf{z}) \leq 1, \quad \forall \mathbf{z} \in \mathcal{C}.
\end{equation}
\end{definition}
\begin{definition}
The vector $\mathbf{z}$ in a convex model~\eqref{eq:convex_model_w_input} is known as the \textit{premise} or \textit{scheduling vector}; it is assumed to be bounded and continuously differentiable in a compact set $\mathcal{C} \subseteq \mathbb{R}^p$ of the scheduling/premise space.
\end{definition}

\begin{definition}
The nonlinear functions $h_i(\mathbf{z}),\ i \in \{1, 2, \ldots, r\}$ in a convex model~\eqref{eq:convex_model_w_input} are known as \textit{\glspl{mf}}; they hold the convex sum property in a compact set $\mathcal{C} \subseteq \mathbb{R}^p$ of the scheduling space.
\end{definition}

\begin{definition}
If $f_i(\mathbf{x}, \mathbf{u})$ are linear, i.e., $f_i(\mathbf{x}, \mathbf{u}) = A_i \mathbf{x} + B_i \mathbf{u}$,~\eqref{eq:convex_model_w_input} is called a \textit{TS} model~\cite{Bernal2022:BOOK}; if, in addition, $\mathbf{z}$ does not include the state $\mathbf{x}$ it is called a \textit{linear polytopic model}~\cite{Takagi1985:FuzzyIdentification}; if the \gls{ts} model comes from fuzzy modeling techniques it is also referred to as \textit{\gls{ts} fuzzy model}.
\end{definition}
% =======================================================================

\begin{definition}
A \textit{\gls{ts} model} is a set of first-order \glspl{ode} whose right-hand side can be expressed as a convex sum of linear models (known as \textit{consequents} or \textit{vertex} models) via nonlinear functions (known as membership functions, \glspl{mf}) that hold the convex sum property \cite{Tanaka2001:FuzzControlSystems}, i.e.
\begin{equation}
\begin{aligned}
    \dot{\mathbf{x}}(t) &= \sum_{i=1}^{r} h_i(\mathbf{z})\big(A_i \mathbf{x}(t) + B_i \mathbf{u}(t)\big), \\ % ", \notag" to not give a tag
    \mathbf{y}(t) &= \sum_{i=1}^{r} h_i(\mathbf{z})\big(C_i \mathbf{x}(t) + D_i \mathbf{u}(t)\big). \label{eq:ts_model}
\end{aligned}
\end{equation}
\end{definition}
where $r \in \mathbb{N}$ is the number of vertex models, $\mathbf{x} \in \mathbb{R}^n$ is the state vector, $\mathbf{u} \in \mathbb{R}^m$ is the input vector, $\mathbf{y} \in \mathbb{R}^q$ is the output vector, $\mathbf{z} \in \mathbb{R}^p$ is the premise vector (which may depend on the state, time, or exogenous parameters and is assumed to be bounded and smooth in a compact set $\mathcal{C} \subseteq \mathbb{R}^p$), $A_i$, $B_i$, $C_i$, $D_i$, $i \in \{1, 2, \ldots, r\}$ are matrices of proper dimensions, and the \glspl{mf} $h_i(\mathbf{z})$, $i \in \{1, 2, \ldots, r\}$ hold the \textit{convex sum property} in $\mathcal{C}$, given by \eqref{eq:convex_sum_property}.
% =======================================================================
% THEOREM - PROOF FOR ANALYSIS PROBLEM
\begin{theorem}
\label{thm:TS_autonomous}
The origin $\mathbf{x} = 0$ of the unforced \gls{ts} model~\eqref{eq:convex_sum_Ai_autonomous} with $\delta \mathbf{x}(t) = \dot{\mathbf{x}}(t)$, $t \in \mathbb{R}$, is asymptotically stable if there exists a matrix $P = {{P}^{\top}} > 0$ such that the following \glspl{lmi} hold:
\begin{equation}
    \label{eq:Vdot_lmi_condition_autonomous}
PA_i + {{A_i}^{\top}} P < 0, \quad i \in \{1, 2, \dots, r\}. 
\end{equation}

\textit{In this case, any trajectory $\mathbf{x}(t)$ starting within the outermost Lyapunov level $V(\mathbf{x}) = {{\mathbf{x}}^{\top}} P \mathbf{x} \leq k$, $k > 0$, in a compact set of the state space $\mathcal{X}$ guaranteeing $\mathbf{z} \in \mathcal{C} \subseteq \mathbb{R}^p$ goes asymptotically to $\mathbf{x} = 0$.}
\end{theorem}

\begin{proof}
Since $P = {{P}^{\top}} > 0$, $V(\mathbf{x}) = \mathbf{x}^{\top} P \mathbf{x}$ is a valid Lyapunov function candidate; its time derivative along the trajectories of the unforced model~\eqref{eq:convex_sum_Ai_autonomous} is

\begin{align}
\dot{V}(t) &= \dot{\mathbf{x}}^{\top} P \mathbf{x} + \mathbf{x}^{\top} P \dot{\mathbf{x}} \\
           &= \left( \sum_{i=1}^{r} h_i(\mathbf{z}) A_i \mathbf{x} \right)^{\top} P \mathbf{x} + {{\mathbf{x}}^{\top}} P \left( \sum_{i=1}^{r} h_i(\mathbf{z}) A_i \mathbf{x} \right) \notag \\
&= \sum_{i=1}^{r} h_i(\mathbf{z}) \mathbf{x}^{\top} \left( A_i^{\top} P + P A_i \right) \mathbf{x}.
\end{align}
where the last equality made use of the property $\sum_{i=1}^{r} h_i(\mathbf{z}) = 1$.
Since $h_i(\mathbf{z}) \geq 0$, $i \in \{1, 2, \ldots, r\}$, $\forall \mathbf{z} \in \mathcal{C}$,
not all simultaneously 0,
it is clear that \glspl{lmi}~\eqref{eq:Vdot_lmi_condition_autonomous} imply $\forall \mathbf{x} \neq 0$ then $\dot{V} < 0$, which means $V(\mathbf{x})$ is a valid Lyapunov function establishing asymptotic stability of $\mathbf{x} = 0$.
There exists a compact set $\mathcal{X}$ of the state space with $0 \in \mathcal{X}$, guaranteeing $\mathbf{z} \in \mathcal{C}$, which means that any trajectory $\mathbf{x}(t)$ in the largest Lyapunov level within $\mathcal{X}$ goes asymptotically to 0.

\end{proof}

% REMARK:

% =======================================================================
% THEOREM - PROOF FOR SYNTHESIS PROBLEM
\begin{remark}
    If the convex sum property holds everywhere, i.e.\ if $\mathbf{z} \in \mathcal{C}$ irrespective of $\mathbf{x}$, the results above are global, i.e.\ $\mathcal{X} = \mathbb{R}^n$.
\end{remark}

% ---- END SECTION-2 ----

% ---- BEGIN SECTION-3 ----
\section{\gls{roa} Computation Using \gls{ts} Method}  \label{sec:ROA_TS}
In this section, the \gls{roa} estimation using \gls{ts} method is described by means of an numerical example.

Consider the dynamical system described by the following equations
\begin{equation}
\begin{aligned}
    \dot{x}_1 &= \ -(x_1)^2 -2(x_2)-2(x_1) \\
    \dot{x}_2 &= \ (x_2)^3-x_2
\end{aligned}
\end{equation}
where an estimate of the \gls{roa} needs to be computed by using \gls{ts}-convex modeling.

Considering the region given by
\begin{equation}
    \mathcal{D} = \left\{ x \in \mathbb{R}^{2} : x_1 \in [-1,+1] , x_2 \in [-0.5,+0.5] \right\}
\end{equation}
the dynamics can be rewritten as
\begin{equation}
    \matrixTwoOne{\dot{x}_1}{\dot{x}_2} = \matrixTwoTwo{-x_1-2}{-2}{0}{(x_2)^2-1)}
\matrixTwoOne{x_1}{x_2}
\end{equation}
and this can be written as
\begin{equation}
\begin{aligned}
\dot{x}(t) &= \sum_{i_1=0}^{1} \sum_{i_2=0}^{1}  
w^{1}_{i_1}(z) w^{2}_{i_2}(z)
\left(
\matrixTwoTwo{-z_{1}^{i_1}-2}{-2}{0}{z_{2}^{i_2}-1)}
% \matrixTwoTwo{-x_1-2}{-2}{0}{(x_2)^2-1)}
\matrixTwoOne{x_1}{x_2}
\right) \\
&= \sum_{\mathbf{i} \in \mathbb{B}^2} \mathbf{w}_{\mathbf{i}}(z) 
\left( A_{\mathbf{i}} x(t)\right)
= A_{\mathbf{w}} x(t),
\end{aligned}
\end{equation}
where the terms are defined as
\begin{equation}
    z_{1}=x_1, \quad z_{2}=(x_2)^2.
\end{equation}
and are defined on the set given by
\begin{equation}
    \mathcal{C}=
    \left\{
    z \in \mathbb{R}^2 : \exists x\in \mathcal{D} \, s.t. \, z_1=x_1,z_2=(x_2)^2   
    \right\}.
\end{equation}
Explicitly, the terms are written as
\begin{equation}
        z_{1} = \sum_{i_1=0}^{1} w^{1}_{i_1}(x) z_{1}^{i_1},\quad z_{2} = \sum_{i_2=0}^{1} w^{2}_{i_2}(x) z_{2}^{i_2}.
\end{equation}
%
% \begin{equation}
% \begin{aligned}
%         z_{1} &= \sum_{i_1=0}^{1} w^{1}_{i_1}(x) z_{1}^{i_1} \\
%         z_{2} &= \sum_{i_2=0}^{1} w^{2}_{i_2}(x) z_{2}^{i_2}
% \end{aligned}
% \end{equation}
where the boundary-terms are computed by
\begin{equation}
    \begin{aligned}
    &\min_{x_1(t),\, x_2(t)} z_1(t) = z_{1}^{i_1=0}, \quad && \max_{x_1(t),\, x_2(t)} z_1(t) = z_{1}^{i_1=1}, \\
    &\min_{x_1(t),\, x_2(t)} z_2(t) = z_{2}^{i_2=0}, \quad && \max_{x_1(t),\, x_2(t)} z_2(t) = z_{2}^{i_2=1}.
    \end{aligned}    
\end{equation}
and the corresponding \glspl{mf} are given by
\begin{equation}
    \begin{aligned}
        w_{i_1=1}^{1}(x)=\frac{z_{1}^{1} - z_1}{z_{1}^{1}-z_{1}^{0}} &, \quad w_{i_1=0}^{1}(x)=1-w_{i_1=1}^{1}(x). \\
        w_{i_2=1}^{2}(x)=\frac{z_{2}^{1} - z_2}{z_{2}^{1}-z_{2}^{0}} &, \quad w_{i_2=0}^{2}(x)=1-w_{i_2=1}^{2}(x).    
    \end{aligned}
\end{equation}
Using Theorem~\ref{thm:TS_autonomous}, the problem can be stated by the following feasibility problem:
\begin{equation}
    \begin{aligned}
    \text{find} \quad & P \\
    \text{subject to} \quad & P  \in \mathbb{S}_{++} \\
    & -\innerS{P A_i}  \in \mathbb{S}_{++},\forall i
    \end{aligned}
\end{equation}
The matrix is computed to be 
\begin{equation}
    P=\matrixTwoTwo{0.2017}{-0.1326}{-0.1326}{0.7656}
\end{equation}
and \gls{roa} is parameterized by a term, $k$ which is computed by
\begin{equation}
    \begin{aligned}
    \label{eq:k_for_RoA_ALGO}
    \text{max} \quad & k \\
    \text{subject to} \quad & \Omega=\left\{ x \in \mathcal{D}: x^{\top}Px \leq k \right\} \\
                     \Omega &\subseteq \mathcal{D}
    \end{aligned}
\end{equation}
and $k$ is computed to be $0.17$
using the method described in \cite{Robles17:Subspace_Based_TS}
and the corresponding \gls{roa} is illustrated in Fig-\ref{fig:RoA_problem_1}.
% PUT THE FIGURE
% CVX_OPT_REPORT_PROBLEM_1_RoA.eps
\begin{figure}[h!]
	\centering
	\includegraphics[width=\linewidth]{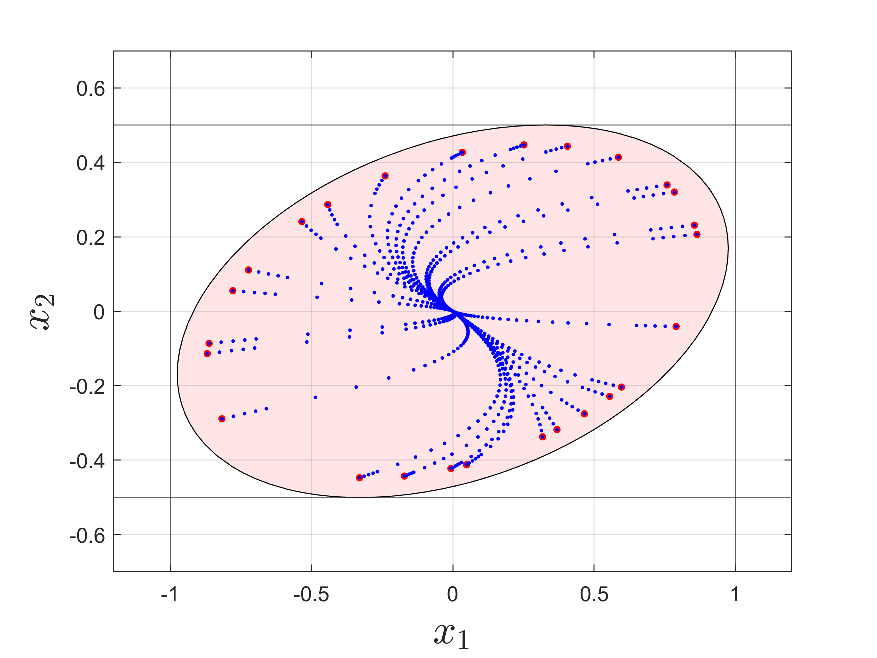}
	\caption{\gls{roa} for the system computed with the \gls{ts} method}
	\label{fig:RoA_problem_1}
\end{figure}      
% ---- END SECTION-3 ----

% ---- BEGIN SECTION-4 ----
\section{\gls{roa} Computation Using Combined State Transformation and \gls{ts} Method}  \label{sec:ROA_CoC}
For the stability analysis method, \gls{ts}-method is useful in that it offers a straightforward algorithm to construct a Lyapunov function to guarantee a local asymptotical stability.
Unfortunately, the results depends on the \glspl{mf} and, as stated previously, that process is not straightforward, since there is not a unique set of \glspl{mf} for a given system.
However, instead of finding the 'best' set of \glspl{mf} to maximize the \gls{roa} , 
we can have a relatively simple 'method for computing the set of  \glspl{mf}' for a given dynamical system and we can get one \gls{roa} ,
and then by implementing a change of variables $\bar{x}=T(x)$ resulting in $\dot{\bar{x}}=\bar{f}(\bar{x})$,
and applying the same \gls{ts}-method to this system results in another \gls{roa} , we can compute the boundary of the ellipsoidal regions (for example by representing them as a set of points that are defined in $\bar{x}$-space)
and transforming them into the $x$-space would give us another region that is guaranteed to be locally asymptotically stable.
Additionally, the aforementioned transformation for the change of variables can be as simple as a linear transformation where $T(x)=T x$, where $T\in\mathbb{R}^{n \times n}$ is any invertible matrix.
For the problem stated in Section \ref{sec:ROA_TS}, this 'change of variables idea' can be implemented.
By choosing a simple transformation matrix
\begin{equation}
    T=\matrixTwoTwo{1}{2}{0}{1}
\end{equation}
which results in the following system
\begin{equation}
\begin{aligned}
    \dot{\bar{x}}_1 &= \ -(\bar{x}_1^2) +4(\bar{x}_1)(\bar{x}_2)-2(\bar{x}_1)+2(\bar{x}_2^3)-4(\bar{x}_2^2)  \\
    \dot{\bar{x}}_2 &= \ (\bar{x}_2^3)-\bar{x}_2 % \bar{x}
\end{aligned}
\end{equation}
and by choosing the following set of \glspl{mf} :
\begin{equation}
    z_1 = \bar{x}_1, \quad z_2 = \bar{x}_2, \quad  z_3 = \bar{x}_2^2.
\end{equation}

% \begin{equation}
% \begin{aligned}
%     z_1 &= \bar{x}_1  \\
%     z_2 &= \bar{x}_2  \\
%     z_3 &= \bar{x}_2^2
% \end{aligned}
% \end{equation}
 By choosing the following region of interest
 \begin{equation}
     \bar{\mathcal{D}} = \left\{  x :  \bar{x}_1 \in [-0.55,+0.55] ,  \bar{x}_2 \in [-0.55,+0.55] \right\}
 \end{equation}
and the \gls{roa} for the $\bar{x}$-domain and the region that is computed by mapping that region back to the original $x$-domain is illustrated in Fig-\ref{fig:RoA_problem_3}.
% PUT THE FIGURE
% CVX_OPT_REPORT_PROBLEM_2_RoA.eps

\begin{figure}[h!]
\vspace{-1.5em}
	\centering
	\includegraphics[width=\linewidth]{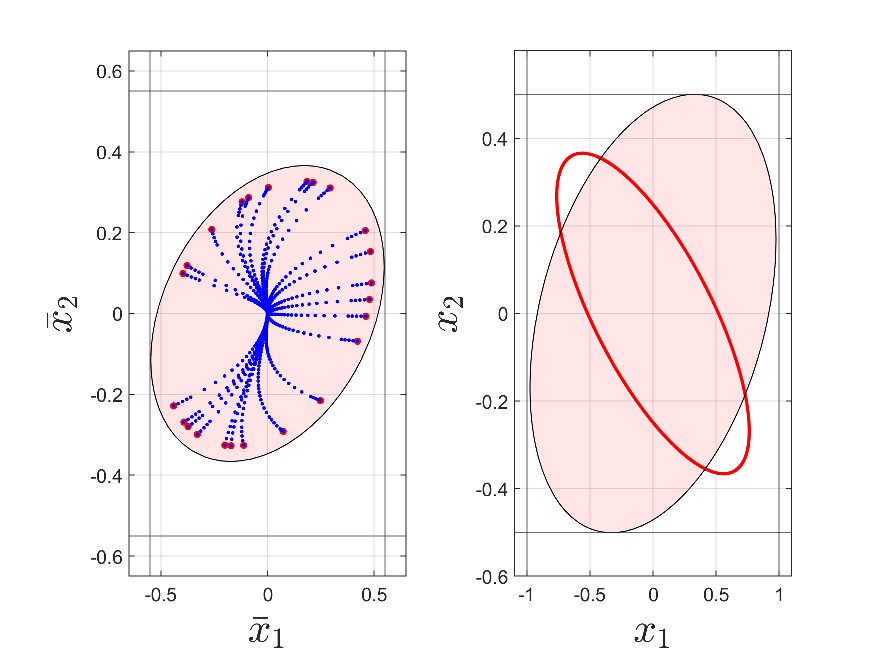}
	\caption{(a) \gls{roa} computed for the system whose state variable is $\bar{x}\in\bar{\mathcal{D}}$
    (b) \gls{roa} computed for the system whose state variable is $x\in\mathcal{D}$ by using the \gls{ts}-method illustrated by red-shaded region and the boundary of the \gls{roa} that is computed in $\bar{\mathcal{D}}$ domain and mapped to $\mathcal{D}$}
	\label{fig:RoA_problem_3}
\end{figure}	

% \begin{figure}[h!]
% 	\centering
% 	\includegraphics[width=\linewidth]{figs/TS_PWC_Fig_1.eps}
% 	\caption{(a) \gls{roa} computed for the system whose state variable is $\bar{x}\in\bar{\mathcal{D}}$
%     (b) \gls{roa} computed for the system whose state variable is $x\in\mathcal{D}$ by using the \gls{ts}-method illustrated by red-shaded region and the boundary of the \gls{roa} that is computed in $\bar{\mathcal{D}}$ domain and mapped to $\mathcal{D}$}
% 	\label{fig:RoA_problem_X}
% \end{figure}

% \begin{figure}[h!]
% 	\centering
% 	\includegraphics[width=\linewidth]{figs/TS_PWC_Fig_2.eps}
% 	\caption{(a) \gls{roa} computed for the system whose state variable is $\bar{x}\in\bar{\mathcal{D}}$
%     (b) \gls{roa} computed for the system whose state variable is $x\in\mathcal{D}$ by using the \gls{ts}-method illustrated by red-shaded region and the boundary of the \gls{roa} that is computed in $\bar{\mathcal{D}}$ domain and mapped to $\mathcal{D}$}
% 	\label{fig:RoA_problem_X}
% \end{figure}

It can be seen that the red-ellipse which represents the boundaries of the ellipsoidal region that is computed using the presented 'change of variables method' has some areas outside of the original \gls{roa} that is computed in the Section \ref{sec:ROA_TS} and that suggests that iteratively using a method (by automatizing the set of  \glspl{mf}  for a given nonlinear dynamical system) results in a larger \gls{roa}.      
% ---- END SECTION-4 ----

% ---- BEGIN SECTION-5 ----
\section{Conclusions}  \label{sec:CONCLUSIONS}
In this study, we explored the use of \gls{sdp} relaxations and
\glspl{lmi} as powerful tools for analyzing and designing control
systems. As a key application, we employed the \gls{ts} modeling
framework to study \gls{roa} estimation problem. The
\gls{ts} approach enabled us to leverage convex optimization techniques to handle nonlinear dynamics effectively. Furthermore,
we proposed an extension to the standard methods, aimed
at enhancing their applicability and reducing conservatism in
practical scenarios. These contributions collectively demonstrate the potential of combining \gls{ts} modeling with \gls{sdp}-based
tools to address challenging problems in nonlinear dynamical systems.      
% ---- END SECTION-5 ----

% \section{\textcolor{red}{TO DO LIST}}
% \begin{table}[htbp]     \label{tab:symbol_definitions}
%     \centering
%     \caption{\textcolor{red}{TO DO LIST}}
%     \begin{tabular}{|p{6cm}|p{1cm}|}
%         \hline
%         \textbf{TASK} & \textbf{TIME[h]} \\
%         \hline
%         \textcolor{red}{check the PRELIMINARIES section (there must be inconsistencies)} & 1h \\
%         \hline
%         \textcolor{red}{check the ACRONYMS (there must be repeating ones)} & 1h \\
%         \hline
%         REFERENCES & 1h \\
%         \hline
%         Journal of Energy Systems CIT & X \\
%         \hline
%         Proofread & 1h \\
%         \hline
%         Submission & 1h \\
%         \hline
%     \end{tabular}
% \end{table}

% \begin{enumerate}
%     \item check the PRELIMINARIES section (there must be inconsistencies)
%     \item check the ACRONYMS (there must be repeating ones)
%     \item add the section ROA-TS
%     \item add the section ROA-CoC
%     \item in ROA-CoC section: add an ALGORITHM
%     \item add the section SIM-RES
%     \item add the section CONCLUSIONS
%     \item REFERENCES
% \end{enumerate}

% \section{Acknowledgments}
% \textcolor{red}{XYZ}

\bibliographystyle{IEEEtran}
\bibliography{Ref.bib}

\end{document}